\documentclass{article}
\usepackage{amsmath,amssymb,amsthm}
\usepackage{graphicx} % Required for inserting images
\usepackage[utf8]{inputenc}
\usepackage{amsmath, amssymb, amsthm}
\usepackage{authblk}
\usepackage{tikz}
\usetikzlibrary{positioning, shapes.geometric, arrows.meta, calc}
\usepackage{hyperref}
\usepackage{enumitem}
\usepackage{geometry}
\usepackage{microtype}
\newtheorem{theorem}{Theorem}[section]

\newtheorem{conjecture}[theorem]{Conjecture}

\newtheorem{proposition}[theorem]{Proposition}
\newtheorem{corollary}[theorem]{Corollary}
\newcommand{\R}{\mathbb R}
\newcommand{\Z}{\mathbb Z}

\DeclareMathOperator{\cone}{cone}

\title{Equitable Subpartitionability Does Not Characterize Box-TDI Polyhedra}

\author[1]{Emiliano Lancini%
  \thanks{\href{mailto:emiliano.lancini@lamsade.dauphine.fr}
  {\nolinkurl{emiliano.lancini@lamsade.dauphine.fr}}---Corresponding Author}}

\author[2]{Francesco Pisanu%
  \thanks{\href{mailto:francesco.pisanu@unive.it}
  {\nolinkurl{francesco.pisanu@unive.it}}}}
  \affil[1]{LAMSADE, CNRS UMR 7243,
Université Paris Dauphine--PSL, Paris, France.}
  \affil[2]{Venice School of Management, Ca’ Foscari University of Venice, Venice, Italy.}

\date{ }
\begin{document}
\maketitle
\vspace{-2em}
\begin{abstract}
Equitable subpartitionability (ESP) is a sufficient condition for an integral system to be box-totally dual integral. We show that the converse fails even when alternative descriptions of the same polyhedron are allowed. We provide a geometrical characterization of polyhedra that can be described by ESP systems in terms of normal cones.  We construct a four-dimensional integer polytope described by a totally equimodular matrix that admits no integral ESP description. This answers negatively a previously posed question by the same authors.
\end{abstract}
\section{Introduction}
% blabla about interaction between systems and polyhedra
% Introduction of Equitable subpartitionability : general concept, time of introduction, what it generalizes.
% Introduction of box-TDI systems, polyhedra, TE polyhedra.
% it is open the question if all box-TDI polyhedra admit a ESP description (as raised by Lancini Pisanu).
% Objective of the paper: answer this question (negatively) and provide a sufficient condition that, toghether with box-TDIness, restore ESP.
%some literature review.
In this work we compare a system and a polyhedral property both of which generalize total unimodularity and yield strong integrality. These two properties are equitable subpartitionability and box-total dual integrality.

A rational system $Ax\le b$, $x\ge0$ is \emph{totally dual integral} (TDI) if the dual of $\max\{c^\top x:Ax\le b,\ x\ge0\}$ admits an integer optimal solution for every integer vector $c$ for which the optimum is finite. When this property is preserved under the addition of arbitrary rational bounds the  variables, we say that the system is \emph{box-totally dual integral}. A polyhedron that can be described by a box-TDI system is called a \emph{box-TDI polyhedron}. We refer to Schrijver~\cite{Schrijver:ILP, schrijver:2003} and to Lancini and Pisanu~\cite{lancini:2026} for a general introduction and an overview of these properties.

A sufficient condition for box-TDIness is \emph{equitable subpartitionability}, abbreviated ESP. This property was introduced by Ding and Zang~\cite{Ding:2002} in 2002 in their study of packing and covering cycles in graphs. It was subsequently used by Chen, Chen, and Zang~\cite{chen2010unified} for box-Mengerian hypergraphs, and extended to general systems of linear inequalities by Ding, Tan, and Zang~\cite{ding2017matching}. Roughly speaking, a system is ESP if every collection of its rows admits two collections satisfying prescribed inequalities on their row sums and their total right-hand side. The two collections need not be subcollections of the original one. This condition generalizes the equitable partition property of totally unimodular matrices. Every ESP system with an integer matrix is box-TDI.

Since equitable subpartitionability implies box-TDIness, one may ask whether every box-TDI polyhedron admits an ESP description, possibly different from its initial description. This question was raised by Lancini and Pisanu~\cite[Question 4.10]{lancini:2026}. 

In this work, we answer this question negatively. We provide a four-dimensional integer polytope which admits a TE description but cannot be described by any ESP system with integer coefficients, and we give a new sufficient condition for a box-TDI polyhedron to admit a ESP description.

%Our approach uses the normal cones at the vertices. We prove that, for full-dimensional rational polytopes contained in the strictly positive orthant, the existence of an integer ESP description is equivalent to the following condition: every integer vector of each vertex normal cone is the sum of two integer vectors of that cone which differ by at most one in each coordinate. We show that this condition can be tested on the subset sums of a Hilbert basis. We then give a subdivision operation which preserves its failure, obtaining integer TE parallelotopes with no integer ESP description in every dimension at least four.

\paragraph{Outline.}
In Section~\ref{sec:preliminaries}, we introduce the definitions and results used throughout the paper. In
Section~\ref{sec:cones}, we characterize the existence of an integer ESP description using normal cones. Section~\ref{sec:construction} presents a box-TDI polyhedron that does not admit an ESP description, together with a general rule for building one.

\section{Preliminaries}\label{sec:preliminaries}
Throughout the paper, all polyhedra are rational and contained in $\R_{>0}$. We denote by $[n]$ the set $\{1, . . . , n\}$.

Consider a linear system $Ax\le b, x\ge0$, where $A$ is an integral matrix and $b$ is rational. Let $R$ be the set of indices of the rows of $A$, and let $a_r$ denote row $r$. A \emph{collection} of elements of $R$ is a subset of elements of $R$, where each element can appear multiple times. For a collection $\Lambda$ of rows of the system $Ax\le b$, we define the vector $d_\Lambda=\sum_{r\in\Lambda}a_r$, and the scalar $b(\Lambda)=\sum_{r\in\Lambda}b_r$.

An \emph{equitable subpartition} of $\Lambda$ is a pair of collections $\Lambda_1,\Lambda_2$ satisfying three conditions:
\begin{align}
b(\Lambda_1)+b(\Lambda_2)&\le b(\Lambda), \tag{E1}\label{eq:E1}\\
d_{\Lambda_1}+d_{\Lambda_2}&\ge d_\Lambda, \tag{E2}\label{eq:E2}\\
d_{\Lambda_1},d_{\Lambda_2}&\ge
\left\lfloor\frac{d_\Lambda}{2}\right\rfloor,
\tag{E3}\label{eq:E3}
\end{align}
Where \eqref{eq:E2} and \eqref{eq:E3} are componentwise.
Remark that the collections $\Lambda_1$ and $\Lambda_2$ need not be subcollections of $\Lambda$. 
The system is \emph{equitably subpartitionable} (\emph{ESP}) if every collection admits an equitable subpartition. 
We say that a polyhedron admits an \emph{ESP description} if there exists an ESP system that describes it.
%Given a collection $\Lambda$, we denote by $\chi^Lambda$ its incidence vector considering the multiplicity.
Remark that, under our assumptions, a polyhedron admits an ESP description only if it is contained in the nonnegative orthant.

The ESP property was used in several papers to prove strong integrality conditions (see for instance \cite{chen2010unified, Ding:2002}). In particular, we have that being ESP is a sufficient condition for box-TDIness.
\begin{theorem}[Ding et al.~\cite{ding2017matching}]\label{thm:ESPimpliesboxTDI}
Every ESP system $Ax\le b, x\ge0$ with $A$ integral is box-TDI.
\end{theorem}
This can be translated in polyhedral terms.
\begin{corollary}
    Every polyhedron that admits a ESP description is a box-TDI polyhedron.
\end{corollary}

Another sufficient condition for a polyhedron to be box-TDI is the existence of a description that uses either totally unimodular or totally equimodular matrices.
An integer matrix is \emph{totally equimodular} (TE) if, for every subset of linearly independent rows, all nonzero maximal minors have the same absolute value; if those values are only in $\{0,\pm1\}$, the matrix is \emph{totally unimodular} (TU). 
\begin{theorem}[Chervet, Grappe, and Robert \cite{Chervet:2021}]\label{thm:TEboxTDI}
    Let $A$ be an integer matrix. Then $A$ is TE if and only if $\{x:Ax\le b\}$ is box-TDI for every rational vector $b$. 
\end{theorem}

Throughout the paper, a \emph{cone} is a polyhedron closed under scalar nonnegative multiplication; in particular, if the cone is pointed, that is it has a vertex, it must be pointed in $0$.
For a polyhedron $P$ and a point $v\in P$, its \emph{normal cone} at $v$
is $N_P(v):=\{z :z^\top (x-v)\le0 $ for every  $x\in P\}$.
Thus, $z\in N_P(v)$ if and only if $v$ maximizes $z^\top x$ over
$P$. For a polytope $P$, we denote its support function by $h_P(z)=\max_{x\in P}z^\top x$.

A \emph{Hilbert basis} of a rational pointed cone $C$ is a finite set of
integer vectors of $C$ which generates $C\cap\Z^n$ by nonnegative
integer combinations.

\section{Balanced Cones}\label{sec:cones}
Let $C$ be a rational cone, we say that the cone is \emph{balanced} if for every integer vector $z$ in $C$, there exist two integer vectors $z_1,z_2$ in $C$ such that:
\begin{itemize}
    \item $z_1+z_2=z$ and
    \item $|z_1-z_2|\le 1$.
\end{itemize}

To test balancedness of a cone, it suffices to consider all vectors obtained by $0,1$ combinations of elements of a Hilbert basis $H$ of the
cone. Note that those combinations can grow superpolynomially in the size of the basis, hence this observation does not yield a polynomial-time test. 

\begin{proposition}\label{prop:necessary}
Let $P$ be a polyhedron that admits an ESP description, and let $v$ be a vertex of $P$ with only strictly positive components. Then $N_P(v)$ is balanced.
\end{proposition}
\begin{proof}
Let $Ax\le b$, $x\ge0$ be an ESP description of $P$, and let $z\in N_P(v)\cap\Z^n$. 
Since the system is TDI by Theorem~\ref{thm:ESPimpliesboxTDI} and $v$ maximizes $z^\top x$, there exists an integer optimal dual solution $y\ge0$. By the strong duality theorem, we have that $ z^\top v=b^\top y\ge(A^\top y)^\top v\ge z^\top v$, thus all members are equal. In particular, since $v$ has only nonzero components, $A^\top y=z$. Moreover, every row used with positive multiplicity in $y$ is tight at $v$.

Let $\Lambda$ be the collection in which row $r$ occurs $y_r$ times.
Then, $d_\Lambda=z$ and $b(\Lambda)=z^\top v$. Let $\Lambda_1,\Lambda_2$ be an equitable subpartition of $\Lambda$, we have that:
\begin{enumerate}
    \item $b(\Lambda) \ge b(\Lambda_1)+b(\Lambda_2)$ for \eqref{eq:E1},
    \item $b(\Lambda_i) \ge d_{\Lambda_i}^\top v$ for $i=1,2$, because $v$ is feasible,
    \item $(d_{\Lambda_1}+d_{\Lambda_2})^\top v\ge d_{\Lambda}^\top v$ for~\eqref{eq:E2}.
\end{enumerate}
Combining these facts together, with the fact that $v$ has only positive components, we have: $z=d_{\Lambda_1}+d_{\Lambda_2}$.
Moreover, since equality holds in point~2, for each $i=1,2$,
$0=b(\Lambda_i)-d_{\Lambda_i}^\top v=\sum_{r\in\Lambda_i}(b_r-a_r^\top v)$.
Therefore, every row appearing in $\Lambda_1$ or $\Lambda_2$ is tight at $v$. Thus, we have that $d_{\Lambda_1}$ and $d_{\Lambda_2}$ are integer vectors of $ N_P(v)$, and \eqref{eq:E3} proves the result.
\end{proof}

Note that nonnegativity of $v$ alone is not sufficient for this argument. If a component of $v$ is zero, then $d_{\Lambda_1}+d_{\Lambda_2}\ge z$ and $(d_{\Lambda_1}+d_{\Lambda_2}-z)^\top v=0$ need not imply equality in that component.

We now prove that if the normal cone of a polyhedron is balanced at every vertex, then we can build an ESP description.

\begin{theorem}\label{thm:sufficient}
Let $P$ be a full-dimensional polyhedron. If every vertex normal
cone of $P$ is balanced, then $P$ admits a finite ESP description. 
\end{theorem}
\begin{proof}
For each vertex $v$ of $P$, let $H_v$ be a Hilbert basis of $N_P(v)$.
Consider the system:
\begin{equation}\label{sys:ESPsystem}
    \left\{\begin{array}{ll}
       a^\top x\le h_P(a)  &  \text{for all } a \in H_v, \quad \text{ for all } v \text{ vertex of } P\\
        x\ge 0 & 
    \end{array}\right.
\end{equation}
By construction, System~\eqref{sys:ESPsystem} describes $P$ and has integer coefficients.

Let $\Lambda$ be a collection of rows of this system. Choose a vertex $v$ maximizing $d_{\Lambda}^\top x$ over $P$. Then, $d_{\Lambda}$ is an integer vector of $N_P(v)$ that is balanced, hence, there exist two integer vectors $z_1,z_2$ in $N_P(v)$ such that $z_1+z_2=d_{\Lambda}$ and $z_1,z_2\ge\lfloor d_{\Lambda}/2\rfloor$.

Since $H_v$ is a Hilbert basis, both $z_1$ and $z_2$ are nonnegative integer
combinations of its elements. Let $\Lambda_1$ and $\Lambda_2$ be the
corresponding collections of rows. Every such row is tight at $v$, so
\[
b(\Lambda_1)+b(\Lambda_2)
=z_1^\top v+z_2^\top v
=h_P(d_{\Lambda})\le b(\Lambda).
\]
Together with the choice of $z_1$ and $z_2$, this proves that
$\Lambda_1,\Lambda_2$ form an equitable subpartition of $\Lambda$.
\end{proof}

The two results above characterize completely polyhedra strictly contained in the positive orthant that admit an ESP description.

\begin{corollary}
    Let $P$ be full-dimensional polyhedron in $\R_{>0}$, then it admits an ESP description if and only if the normal cone on each of its vertices is balanced.
\end{corollary}

\section{The Counterexample}\label{sec:construction}
The results in Section~\ref{sec:cones} yield to a building rule for generating box-TDI polyhedra that do not admit a ESP description.

\begin{theorem}\label{thm:construction}
Let $G$ be a non-singular $n\times n$ TE matrix.
Suppose that $C=\cone(\text{rows of }G)$  is not balanced. Let $t$ be a positive integer such that $tG^{-1}$ is integral, and denote by $T$ the uniform vector of components $t$. Let $v\in\Z^n$ satisfy
\[
v_i\ge1+\sum_{j\in [n]}\max\{(tG^{-1})_{ij},0\}\qquad \text{ for all }i\in[n].
\]
Then, $P=\{x:-T\le G(x-v)\le0\}$ is an integer box-TDI polytope contained in $\R_{>0}^n$ that does not admit an ESP description.
\end{theorem}
\begin{proof}
The constraint matrix of $P$ is totally equimodular, hence $P$ is box-TDI by~\ref{thm:TEboxTDI}. Since $G$ is nonsingular, the polytope is bounded and its vertices are $ v-tG^{-1}z$ for all $\{0,1\}$ vectors $z$. In particular, $x\ge 1$ for all $x$ in $P$.
At the vertex $v$, we have $N_{P}(v)=C$. Since $C$ is not balanced, $P$ does not admit a ESP description.
\end{proof}

Theorem~\ref{thm:construction} does not fully answer our investigation, since we must still prove that there exists at least one matrix satisfying the hypotheses of the theorem.
The smallest matrix we were able to find is:
\[
\Gamma=
\begin{pmatrix}
1&0&0&1\\
0&1&0&1\\
0&0&1&1\\
1&1&1&1
\end{pmatrix}.\]
$\Gamma$ is a well-known minimally non-TU matrix, in particular, it is TE.
We proceed by showing that the cone generated by the rows of $\Gamma$ is not balanced.
\begin{proposition}\label{prp:42}
    The cone $C $ generated by the rows of $\Gamma$, $\gamma_1,\dots\gamma_4$, is not balanced.
\end{proposition}
\begin{proof}
    Consider the vector $z=(1,1,1,2)$, this vector belong to $C$ because $z=(\gamma_1+\gamma_2+\gamma_3+\gamma_4)/2$. 
    In any balanced decomposition $z=x_1+x_2$, both $x_1$ and $x_2$ have last component equal to $1$ and first three components in $\{0,1\}$. Since $C\cap\{x_4=1\}=\operatorname{conv}\{\gamma_1,\dots,\gamma_4\}$, the only such integer points are $\gamma_1,\dots,\gamma_4$.
However, $z$ cannot be the sum of two rows of $\Gamma$: each row has an odd sum of its first three components, and that sum for $z$ is $3$. Thus $C$ is not balanced.
\end{proof}
\section{Conclusions}\label{sec:conc}
In this paper we answered negatively to a question raised in~\cite{lancini:2026}: there exist box-TDI polyhedra that do not admit an ESP description. In order to do so, we translated the ESP property in to a geometrical property on the normal cones of a polyhedron, the balancedness.
We both gave an example of such a polyhedron and a tool to generate an infinite family of such examples. In particular, our example can be used to build polyhedra of every dimension $\ge 4$ that are box-TDI and do not admit an ESP description.
Recognizing box-TDI polyhedra is co-NP-complete \cite{chervet:2023:hard}. This motivates the following
conjecture, which we leave open.
\begin{conjecture}
    Testing whether a cone is balanced is Co-NP-complete.
\end{conjecture}
\subsection*{Acknowledgment of AI usage}
 The mathematical results and proofs in this paper were obtained independently by the authors, the counterexample in Proposition~\ref{prp:42} was obtained by an AI. Generative AI was also used to assist with the language editing. All mathematical claims, citations, and references were independently checked by the authors, who take full responsibility for the content of the paper.

\bibliographystyle{plain}
\bibliography{ref}
\end{document}